\documentclass{llncs}
\usepackage{makeidx}  % allows for indexgeneration
\usepackage{graphicx}
\usepackage{subfigure}
\usepackage{amsfonts}
\usepackage{amsmath}
\usepackage{algorithm}
\usepackage{algorithmic}
\usepackage{multirow}
\usepackage{times}
\usepackage{enumitem}
\usepackage{hyperref}

\usepackage{booktabs}
\usepackage{url}

\begin{document}
\mainmatter              % start of the contributions
\title{Single-Pass PCA of Large High-Dimensional Data
\thanks{to appear in Proc. IJCAI 2017.}
}
%
%\titlerunning{Hamiltonian Mechanics}  % abbreviated title (for running head)
\author{Wenjian Yu\inst{1} \and Yu Gu\inst{2} \and
Jian Li\inst{3} \and Shenghua Liu\inst{4} \and Yaohang Li\inst{5}}
%
%\authorrunning{Ivar Ekeland et al.} % abbreviated author list (for running head)
%Craig Chambers, Kim B. Bruce, and Elisa Bertino}
%
\institute{Department of Computer Science and Technology, Tsinghua National Lab of Information Science and Technology, Tsinghua University, Beijing, China,\\
\email{yu-wj@tsinghua.edu.cn}
\and
Institute for Interdisciplinary Information Sciences, Tsinghua University, Beijing, China,\\
\email{guyu13@mails.tsinghua.edu.cn}
\and
Department of Electronic Engineering, Tsinghua University, Beijing, China,\\
\email{j-l14@mails.tsinghua.edu.cn}
\and
Institute of Computing Technology, Chinese Academy of Sciences, Beijing, China, \\
\email{liushenghua@ict.ac.cn}
\and
Department of Computer Science, Old Dominion University, Norfolk, VA 23529, USA, \\
\email{yaohang@cs.odu.edu} }

\maketitle              % typeset the title of the contribution

\begin{abstract}
Principal component analysis (PCA) is a fundamental dimension reduction tool in statistics and machine learning. For large and high-dimensional data, computing the PCA (i.e., the singular vectors corresponding to a number of dominant singular values of the data matrix) becomes a challenging task. In this work, a single-pass randomized algorithm is proposed to compute PCA with only one pass over the data. It is suitable for processing extremely large and high-dimensional data stored in slow memory (hard disk) or the data generated in a streaming fashion. Experiments with synthetic and real data validate the algorithm's accuracy, which has orders of magnitude smaller error than an existing single-pass algorithm. For a set of high-dimensional data stored as a 150 GB file, the proposed algorithm is able to compute the first 50 principal components in just 24 minutes on a typical 24-core computer, with less than 1 GB memory cost.
\end{abstract}
\section{Introduction}

Many existing machine learning models, no matter supervised or
unsupervised, rely on dimension reduction of input data. 
Even the applications of Deep Neural Networks on natural 
language processing tasks~\cite{bahdanau2014neural,kim2014convolutional}, 
prefer to use an embedding of each word 
in a sentence~\cite{bengio2003neural,mikolov2013distributed}, which essentially reduces the data dimensionality.
Principal component analysis (PCA) is an efficient and 
well-structured dimension reduction technique~\cite{Halko2011,Learning2001}.
However, how to calculate PCA of large-size (say terabyte) and 
high-dimensional dense data in a limited-memory computation node
is still an open problem.
Plus some data are generated in stream, e.g., from internet traffic,
and signals from internet of things, we need a 
kind of pass-efficient algorithm, or even
single-pass algorithm, to realize the dimension reduction
of input data.

%Matrix is universally used to express various large data sets nowadays. Therefore, many computational problems on such data are reduced to matrix computations including matrix multiplication, matrix factorizations, low-rank matrix approximation, and $l_2$-regression, etc. Some data are too huge to be stored in RAM or are generated sequentially (e.g., from internet traffic or time series data). They request a kind of pass-efficient algorithm, or even single-pass algorithm, to realize the matrix operations. A single-pass algorithm requires only one pass over the data, and is particularly useful and efficient for the data generated in a streaming fashion or stored 
%in slow memory \cite{Halko2011review,ZhangZ2016}. 
%% how to add this idea? 
%It also allows the computation with small or fixed RAM size \cite{TRIEST2016}.
%
%Principal component analysis (PCA) is a fundamental dimensionality reduction model in statistics and machine learning \cite{Halko2011,Learning2001}. It is the basis of many techniques in data mining, pattern recognition and information retrieval.
%, e.g., the latent semantic analysis of large database of documents \cite{matrix2007}. 
%Computing the PCA of a data set is equivalent to computing a truncated singular value decomposition (SVD) of the data matrix. % (possibly after suitable normalization).

A single-pass algorithm has the benefit of requiring only one pass over the data, and 
is particularly useful and efficient for streaming data or data stored in slow
memory~\cite{Halko2011review,ZhangZ2016}. 
It also allows the computation with small or fixed RAM size \cite{TRIEST2016}. Although there are provable single-pass truncated SVD algorithms for symmetric positive semi-definite (SPSD) matrices \cite{drineas2005nystrom,wang2013improving,gittens2013revisiting,wang2016towards}, the study for more general matrices is not sufficient.
\cite{Halko2011review} proposed a single-pass algorithm for
approximately calculating SVD for general matrices, but with a significant cost of  
accuracy.
\cite{LargePCA2014} developed a PCA algorithm 
for large-size data, but only applicable to low-dimensional data (less than
one thousand in dimension). 
%The algorithm first constructs the correlation
%matrix (which is of small size) through a single pass over data, and then
%computes eigen-decomposition of the correlation matrix. Obviously, it becomes
%infeasible if the dimension of data increases to the size of data set. 
A recent single-pass algorithm was proposed for the PCA of matrix products~\cite{singlePCA_NIPS16}, 
which is a generalization of computing PCA of a matrix.
However, the algorithm also assumes that the data has a small dimension.  
Frequent-directions (FD) algorithm~\cite{Liberty2013} was
a single-pass and deterministic matrix sketching scheme \cite{woodruff2014sketching}, which is 
useful for matrix multiplication problem~\cite{ZhangZ2016}, but the
PCA computation.

Randomized matrix computation has gained significant increases in popularity as the data sets are becoming larger and larger~\cite{mahoney2011randomized}. It has been revealed that randomization can be a powerful computational resource for developing algorithms with improved runtime and stability properties \cite{Halko2011review,CACM2016,ZhangW2016,wangss2015}. Compared with classic algorithms, the randomized algorithm involves the same or fewer floating-point operations (\emph{flops}), and is more efficient for truely large high-dimensional data sets, by exploiting modern computing architectures.
%In this work, we follow the idea in randomized matrix algorithms \cite{Halko2011review,Martin2016,CACM2016}, and propose a fast single-pass algorithm for the PCA of large data sets. 
An idea of randomization is using random projection to identify the subspace capturing the dominant actions of a matrix $\mathbf{A}$. With the subspace's orthonormal basis matrix $\mathbf{Q}$, a so-called QB approximation is obtained: $\mathbf{A} \approx \mathbf{QB}$. 
This produces a smaller sketch matrix $\mathbf{B}$, and facilitates the computation of near-optimal decompositions of $\mathbf{A}$. 
%Performing standard SVD on the small matrix $\mathbf{B}$ results in the decompositions of $\mathbf{A}$. 
A simple implementation of this idea and related techniques and theories have been presented in \cite{Halko2011review}.
%, demonstrating its effectiveness in producing near-optimal low-rank matrix approximations. 
With the merit of requiring a small constant number of passes over data, this algorithm has been applied to compute PCA of data sets that are too large to be stored in RAM \cite{Halko2011}. It has also been employed to speed up the distributed PCA, without compromising the quality of the solution \cite{Woodruff2014}.
However, this basic \texttt{randQB} algorithm still involves several passes instead of a single pass over data, which makes it not efficient enough or infeasible for some situations.  

%Several single-pass matrix algorithms have been proposed recently. 
%In \cite{Halko2011review}, a single-pass algorithm was proposed for approximately calculating SVD. However, it degrades the accuracy largely. In \cite{LargePCA2014}, the PCA for large data sets with data dimension less than one thousand  was investigated. The algorithm first constructs the correlation matrix (which is of small size) through a single pass over data, and then computes eigen-decomposition of the correlation matrix. Obviously, it becomes infeasible if the dimension of data increases to the size of data set. A single-pass algorithm was recently proposed for the PCA of matrix products \cite{SinglePCA_NIPS16}, a generalization of computing PCA of a matrix. However, the algorithm also considers the setting that data dimension is much smaller than data size, and is not suitable for processing high-dimensional data.   
 %Another single-pass algorithm is the frequent-directions (FD) algorithm \cite{Liberty2013}, which is a deterministic matrix sketching scheme.
 %It can be applied to the matrix multiplication problem \cite{ZhangZ2016}, instead of computing PCA.
%  Based on it a single-pass approximate matrix multiplication algorithm has been developed \cite{ZhangZ2016}. However, the algorithm is not applicable to the computation of PCA.

Progress has also been achieved based on the randomized algorithm for QB approximation. In \cite{Mary2015}, the basic \texttt{randQB} algorithm \cite{Halko2011review} was slightly modified for computing the QR factorization. The main efforts were paid to investigate the algorithm's performance scaling on shared-memory multi-core CPUs with multiple GPUs, and the comparison with the traditional QR factorization with column pivoting (QRCP). The results demonstrated that the randomized algorithm could be an excellent computational tool for many applications, with growing potential on the emerging parallel computers. 
In \cite{Martin2016}, a randomized blocked algorithm was proposed for computing rank-revealing factorizations in an incremental manner. Although it enables adaptive rank determination, the algorithm needs to access the matrix for a number of times and is not efficient for large-size data. 
%These drawbacks are overcome by the techniques proposed in . With an economic error indicator for the auto-rank problem of low-rank matrix approximation, the algorithm becomes adaptive to large, sparse matrix, and achieves high efficiency on the modern multi-core computer. In , a reorganization of the algorithm which moves the matrix multiplications out of the loop for obtaining more performance gain from BLAS-3 operation was also presented. It enables a pass-efficient variant of the algorithm. 

%This work is inspired by \cite{Halko2011} and \cite{Martin2016}. 
Therefore, we reconstruct the randomized blocked algorithm \cite{Martin2016} and enforce numerical guarantee for the algorithm robustness as well, which results in a single-pass PCA algorithm owning the following advantages. 
\begin{itemize}[leftmargin=*]
\setlength{\itemsep}{0pt}
\setlength{\parsep}{0pt}
\setlength{\parskip}{0pt}
            \item \emph{Single-pass}: it involves only one pass over specified large high-dimensional data.
            \item \emph{Efficiency}: it has $\mathcal{O}(mnk)$ or $\mathcal{O}(mn\log (k))$ time complexity and $\mathcal{O}(k(m+n))$ space complexity for computing $k$ principal components of an $m\times n$ matrix data, and well adapts to parallel computing.
            \item \emph{Accuracy}: it has a theoretical error bound, and empirically shows much less error than the single-pass algorithm in \cite{Halko2011review}, offering good PCA accuracy for matrices
    with different distributions of singular values. 
\end{itemize}

%Due to the gain from BLAS-3 operation, the performance of the proposed algorithm is the same or better than the existing randomized algorithms, and re-owns the pass-efficient property. 
We have examined the effectiveness of the proposed single-pass algorithm for
performing PCA on large-size ($\sim$150 GB) dense data with high dimension, which cannot be fit in RAM (32 GB). % on a typical multi-core computer. 
The experimental results show that our single-pass algorithm 
outperforms the standard SVD and existing competitors, by 
significantly reduced time and memory usage, and accuracy guarantees. For reproducibility, the codes of the proposed algorithm and programs for the experiments in Section 4 will be shared on \url{https://github.com/WenjianYu/rSVD-single-pass}.

% The rest of the paper is organized as follows. We first define the notation used in this paper and introduce the background of the randomized algorithm for low-rank matrix approximation. Then, we describe the new single-pass algorithm for computing truncated SVD and provide theoretic analysis. Finally, we provide experiment results, which demonstrate the accuracy and efficiency of the presented algorithm for handling large data sets. 

\section{Preliminaries}
We assume that all matrices considered in this work are real valued, although the generalization to complex-valued matrices is of no difficulty. An orthonormal matrix denotes a matrix whose columns are a set of orthonormal vectors; 
$\mathbf{I}$ denotes the identity matrix.
And, we follow the Matlab convention for specifying row/column indices of a matrix.

\subsection{Singular Value Decomposition and PCA}
Let $\mathbf{A}$ denote an $m \times n$ matrix. The SVD of $\mathbf{A}$ is 
\begin{equation}\label{eq:svd}
\mathbf{A = U \Sigma V}^{\top},
\end{equation}
where $\mathbf{U}$ and $\mathbf{V}$ are $m \times $min($m,n$) and $n \times $min($m,n$) orthonormal matrices respectively, and $\mathbf{\Sigma}$ is a diagonal matrix. The diagonal entries of $\mathbf{\Sigma}$ are the descending singular values of $\mathbf{A}$: $\sigma_{11} \ge \sigma_{22} \ge \cdots \ge 0$. 
The columns of matrices $\mathbf{U}$ and $\mathbf{V}$ are the left and right singular vectors, respectively.

Taking the first $k,~k<\min(m,n)$, columns of  $\mathbf{U}$ and $\mathbf{V}$ respectively, and the first $k$ singular values in $\mathbf{\Sigma}$, we have the truncated (partial) SVD of matrix $\mathbf{A}$:
\begin{equation}\label{eq:psvd}
\mathbf{A}_k = \mathbf{U}_k \mathbf{\Sigma}_k \mathbf{V}_k^{\top} ,
\end{equation}
where $\mathbf{U}_k$ and $\mathbf{V}_k$ include the first $k$ columns of $\mathbf{U}$ and $\mathbf{V}$, respectively. 
$\mathbf{\Sigma}_k$ is the $k\times k$ up-left submatrix of $\mathbf{\Sigma}$. $\mathbf{A}_k$ is actually the optimal rank-$k$ approximation of $\mathbf{A}$, in terms of $l_2$-norm and Frobenius norm \cite{EY1936}.

The approximation properties of the SVD explain the equivalence between SVD and PCA. Suppose each row of matrix $\mathbf{A}$ is an observed data. The matrix is assumed to be centered, i.e., the mean of each column is equal to zero. Then, the 
%right singular vectors $\mathbf{v}_i$ are called principal component directions of $\mathbf{A}$. The left
leading right singular vectors $\{\mathbf{v}_i\}$ of $\mathbf{A}$ are the %normalized
principal components. Particularly, $\mathbf{v}_1$ is the first principal component.
% \cite{matrix2007}.  

\subsection{The Basic Randomized Algorithm for PCA}
The algorithm in \cite{Halko2011} is based on the basic \texttt{randQB} algorithm for QB approximation, and described as Algorithm 1. $\mathbf{\Omega}$ is a Gaussian i.i.d. matrix. Replacing it with a \emph{structured} random matrix is also feasible, and can reduces the computational cost for a dense $\mathbf{A}$ from $\mathcal{O}(mnl)$ to $\mathcal{O}(mn\log(l))$ flops \cite{Halko2011review}.  
%Although other random matrices might work equally well, the choice of the Gaussian matrix provides some theoretical and practical advantages \cite{Halko2011review,Gu2015}. 
The over-sampling technique which uses $\mathbf{\Omega}$ with more than $k$ columns is employed for better accuracy \cite{Halko2011review}. Usually, the over-sampling parameter $s$ is a small integer, like 5 or 10. ``orth($\mathbf{X}$)'' denotes the orthonormalization of the columns of $\mathbf{X}$. In practice, it is achieved efficiently by a call to a packaged QR factorization (e.g., \texttt{qr(X, 0)} in Matlab), which implements the QR factorization without pivoting. 
%``svd'' stands for a standard LAPACK routine for computing SVD.
        \begin{algorithm}
        \caption{Basic randomized scheme for truncated SVD}
        \label{alg1}
        \begin{algorithmic}[1]
            \REQUIRE $\mathbf{A}\in \mathbb{R}^{m\times n}$, rank $k$, over-sampling parameter $s$.
            \STATE $l=k+s$;
            \STATE $\mathbf{\Omega}=$  randn($n, l$);
            \STATE $\mathbf{Q}=$ orth$(\mathbf{A\Omega})$;
            \STATE $\mathbf{B}= \mathbf{Q}^{\top}\mathbf{A}$;
            \STATE $\mathbf{[\tilde{U}, S, V]}=$ svd($\mathbf{B}$);
        	  \STATE $\mathbf{U= Q\tilde{U}}$;
        	  \STATE $\mathbf{U}=\mathbf{U}(:, 1:k)$; $\mathbf{V}=\mathbf{V}(:, 1:k)$; $\mathbf{S}= \mathbf{S}(1:k,1:k)$;
        	  \STATE \textbf{return} $\mathbf{U, S, V}$.
        \end{algorithmic}
        \end{algorithm}
        
The first four steps in Algorithm 1 is the basic \texttt{randQB} scheme for building $\mathbf{A}$'s QB approximation. This procedure could not produce the optimal low-rank approximation. % as the approach based on the SVD of $\mathbf{A}$. 
However, in many applications the optimal approximation is not necessary, and even impossible to obtain due to the high computational complexity of performing SVD. The existing work has revealed that this randomized algorithm often produces a good enough solution.
Compared with the classic rank-revealing QR factorization \cite{MatrixBook} for low-rank approximation, it has less computational cost and can obtain substantial speedup on a parallel computing platform \cite{Martin2016}.

The error of the randomized QB approximation could be large for the matrix whose singular values decay slowly \cite{Halko2011review}. This can be eased by a technique called \emph{power scheme} \cite{Rokhlin2009}. It is based on the fact that matrix $(\mathbf{A}\mathbf{A}^{\top})^P\mathbf{A}$ has exactly the same singular vectors as $\mathbf{A}$, but its $j$-th singular value is $\sigma_{jj}^{2P+1}$. This largely reduces the relative weight of the tail singular values.  Thus, performing the randomized QB procedure on $(\mathbf{A}\mathbf{A}^{\top})^P\mathbf{A}$ can achieve more accurate approximation. More theoretical analysis can be found in Sec. 10.4 of \cite{Halko2011review}. On the other hand, the power scheme increases the number of passes over $\mathbf{A}$ from $2$ to $2P\!+\!2$ \cite{Halko2011review}.

It should be mentioned that the output of the randomized approximation algorithms is a random variable, as it depends on the drawing of a random matrix. However, it has been proven that the variation of this random variable is small, which is called  \emph{the effect of concentration in measure}, suggesting that for practical purpose the algorithm is deterministic \cite{Candes2015Algori}.
%means the output is always very close to the variable's expectation. 
For details, please refer to \cite{Halko2011review}.
%The bound of the expectation of the approximation error and other related theoretical results have been derived in \cite{Halko2011review}. % More experimental results were presented in \cite{Martin2016}. 

\subsection{An Existing Single-Pass Algorithm}
Algorithm 1 still involves two passes over matrix $\mathbf{A}$. It is not favorable for the large or streaming data. In  \cite{Halko2011review}, a single-pass algorithm was proposed as a remedy. It is described as Algorithm 2, where only Step 2 needs the access of matrix $\mathbf{A}$.

        \begin{algorithm}
        \caption{An existing single-pass algorithm}
        \label{alg2}
        \begin{algorithmic}[1]
            \REQUIRE $\mathbf{A}\in \mathbb{R}^{m\times n}$, rank parameter $k$.
            \STATE Generate random $n\times k$ matrix $\mathbf{\Omega}$ and $m\times k$ matrix $\tilde{\mathbf{\Omega}}$;
            \STATE Compute $\mathbf{Y}\!=\!\mathbf{A\Omega}$ and $\tilde{\mathbf{Y}}\!=\!\mathbf{A}^{\top}\tilde{\mathbf{\Omega}}$ in a single pass over $\mathbf{A}$;
            \STATE $\mathbf{Q}=$ orth$(\mathbf{Y})$;  $\tilde{\mathbf{Q}}=$ orth$(\tilde{\mathbf{Y}})$; 
            \STATE Solve linear equation $\tilde{\mathbf{\Omega}}^{\top}\mathbf{QB}= \tilde{\mathbf{Y}}^{\top}\tilde{\mathbf{Q}}$ for $\mathbf{B}$;
            \STATE $[\tilde{\mathbf{U}}, \mathbf{S}, \tilde{\mathbf{V}}]=$ svd($\mathbf{B}$);
        	  \STATE $\mathbf{U= Q}\tilde{\mathbf{U}}$; $\mathbf{V= }\tilde{\mathbf{Q}}\tilde{\mathbf{V}}$;
        	  \STATE \textbf{return} $\mathbf{U, S, V}$.
        \end{algorithmic}
        \end{algorithm}

Step 3 of Algorithm 2 results in matrices $\mathbf{Q}$ and $\mathbf{\tilde{Q}}$ such that $\mathbf{A} \approx \mathbf{QQ}^{\top} \mathbf{A} \tilde{\mathbf{Q}}\tilde{\mathbf{Q}}^{\top}$. Then, the problem becomes how to compute the small matrix $\mathbf{B}= \mathbf{Q}^{\top}\mathbf{A} \tilde{\mathbf{Q}}$. One can find out that
\begin{equation}\label{eq:appox1}
\tilde{\mathbf{Q}}^{\top}\tilde{\mathbf{Y}} = \tilde{\mathbf{Q}}^{\top}\mathbf{A}^{\top}\tilde{\mathbf{\Omega}} \approx 
\tilde{\mathbf{Q}}^{\top}\mathbf{A}^{\top}\mathbf{Q}\mathbf{Q}^{\top}\tilde{\mathbf{\Omega}} = \mathbf{B}^{\top}\mathbf{Q}^{\top} \tilde{\mathbf{\Omega}},
\end{equation}
So, $\mathbf{B}$ is approximately computed in Step 4. Because there are two or more approximations in the deduction, the accuracy of this algorithm or its variants in \cite{Halko2011review} is not good. We will reveal this through experiments.

\subsection{The Randomized Blocked Algorithm}
The randomized blocked algorithm in \cite{Martin2016} is inspired by a greedy  Gram-Schmidt procedure for the orthonormalization step of basic \texttt{randQB}, which constitutes an iterative procedure with the error of the QB approximation updated. Then, the algorithm  is converted to a blocked version to attain high performance of linear algebraic computation (see Fig. 1). %Note that matrix $\mathbf{A}$ is overwritten in Step (6).  
It is easy to prove that, if the algorithm is executed in exact arithmetic $\mathbf{Q}$ is orthonormal, $\mathbf{B}=\mathbf{Q}^{\top} \mathbf{A}$, and
after Step (6) $\mathbf{A}$ becomes the approximation error: $\mathbf{A-QB}$. 

\begin{figure}[h]
\centering %\setlength{\fboxrule}{0.5pt}
%\setlength{\fboxsep}{0.5pt}
%\fbox{
\includegraphics[height= 1.7in]{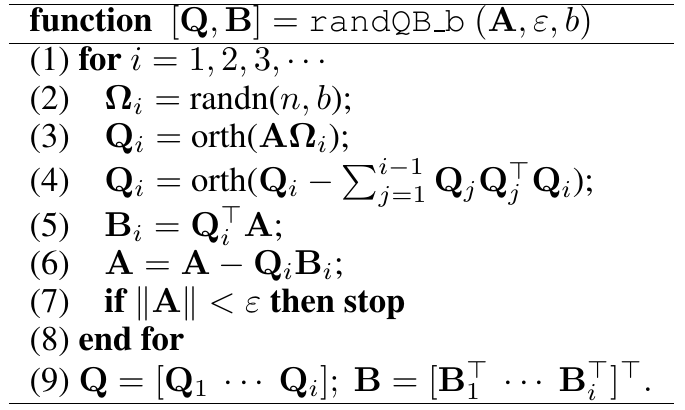}
%}
\caption{The blocked randQB algorithm, from \cite{Martin2016}.}
%\caption{The blocked randQB algorithm, from Martin2016. Superscript * means transpose.}
\label{fig1}
\end{figure}

The algorithm is mathematically equivalent to the basic \texttt{randQB} procedure (the first 4 steps of Algorithm 1), except for the stopping criterion.  
Notice that the re-orthogonalization step (4) is for easing the accumulation of numerical round-off error. 
Experiments in \cite{Martin2016} showed that it has the same or better accuracy than the column-pivoted QR factorization, and runs much faster on multi-core architectures.

% for drawing Fig. 1.
%\begin{center}
%\begin{tabular}{l}
%\hline
%\textbf{function}~ $[\mathbf{Q}, \mathbf{B}]= $ \texttt{randQB\_b} $(\mathbf{A}, \varepsilon, b)$  \\ 
%\hline
%%\textbf{Input}: $\boldsymbol{A},~ \varepsilon,~ b$  \\
%%\textbf{Ouput}: $\boldsymbol{Q},~ \boldsymbol{B}~$. \\
%(1) \textbf{for} $i=1,2,3,\cdots$ \\
%(2)  ~ ~$ \boldsymbol{\Omega}_i=$ randn($n, b$); \\
%(3)  ~ ~$\mathbf{Q}_i=$  orth($\mathbf{A}\mathbf{\Omega}_i$); \\
%(4)  ~ ~$\mathbf{Q}_i=$  orth($\mathbf{Q}_i- \begin{matrix} \sum_{j=1}^{i-1} \mathbf{Q}_j\mathbf{Q}_j^{\top}\mathbf{Q}_i \end{matrix}$); \\
%(5)  ~ ~$\mathbf{B}_i=\mathbf{Q}_i^{\top}\mathbf{A}$; \\
%(6)  ~ ~$\mathbf{A}= \mathbf{A}-\mathbf{Q}_i\mathbf{B}_i$; \\
%(7)  ~ ~\textbf{if} $\|\mathbf{A}\| < \varepsilon$ \textbf{then stop} \\
%(8) \textbf{end for} \\
%(9) $\mathbf{Q}= [\mathbf{Q}_1~\cdots~\mathbf{Q}_i];~ \mathbf{B}= [\mathbf{B}_1^{\top}~\cdots~\mathbf{B}_i^{\top}]^{\top}. $ \\
%\hline
%\end{tabular}
%\end{center}

\section{Methodology}
In this section, we first propose a pass-efficient blocked algorithm, and then enhance its robustness with the re-orthonormalization. Finally, we present the single-pass algorithm for computing PCA.

\subsection{A Pass-Efficient Blocked Algorithm}
The blocked randQB procedure in Fig. 1 facilitates adaptive rank determination, but increases the number of passes over the data matrix. For many scenarios of using PCA, the rank parameter $k$ is a known value. Otherwise, it is often referred to as the fixed-precision problem \cite{Halko2011review}. The algorithm proposed in this work can be extended to handle the fixed-precision problem, which is however not the main focus of this paper.

Without the evaluation of approximation error, the block randQB procedure can be modified to be a pass-efficient procedure. It is presented as Algorithm 3, with the re-orthogonalization step ignored. For simplicity, we also assume that $k$ is a multiple of $b$.
        \begin{algorithm}
        \caption{A pass-efficient blocked algorithm}
        \label{alg3}
        \begin{algorithmic}[1]
            \REQUIRE $\mathbf{A}\in \mathbb{R}^{m\times n}$, rank parameter $k$, block size $b$.
            \STATE $\mathbf{Q} = [~]; ~ \mathbf{B} = [~];$
            \STATE $\mathbf{\Omega}= $ randn($n, k$);
        \STATE $\mathbf{G} = \mathbf{A\Omega}$;
        \STATE $\mathbf{H} = \mathbf{A}^{\top}\mathbf{G}$;
            \FOR{$i=1, 2, \cdots, k/b$}
            %\FOR{$\sqrt{E} > \varepsilon$}
            \STATE $\mathbf{\Omega}_i = \mathbf{\Omega}(:,~ (i-1)b+1:ib)$;
            \STATE $\mathbf{Y}_i = \mathbf{G}(:,~ (i-1)b+1:ib) - \mathbf{Q}(\mathbf{B\Omega}_i)$;
			\STATE $[\mathbf{Q}_i,~\mathbf{R}_i]= $ qr($\mathbf{Y}_i$);
			\STATE $\mathbf{B}_i= \mathbf{R}_i^{-\top}(\mathbf{H}(:,~ (i-1)b+1:ib)^{\top}- \mathbf{\Omega}_i^{\top}\mathbf{B}^{\top}\mathbf{B}) $;
            \STATE $\mathbf{Q} = [\mathbf{Q}, ~\mathbf{Q}_i]$;~ $\mathbf{B} = [\mathbf{B}^{\top},  ~\mathbf{B}_i^{\top}]^{\top}$;
%            \STATE $\mathbf{B} = \left[ \begin{array}{c}
%                 \mathbf{B}\\
%                 \mathbf{B}_i
%            \end{array}\right]$;
            \ENDFOR
        
        \end{algorithmic}
        \end{algorithm}

A major difference between Algorithm 3 and the algorithm in Fig. 1 is that the multiplications with $\mathbf{A}$ is moved out of the loop. With $\mathbf{G=A\Omega}$, Steps 7 and 8 in Algorithm 3 perform the same function as Step (3) in the latter. Here, ``qr'' denotes a standard QR factorization, i.e., $\mathbf{Q}_i\mathbf{R}_i=\mathbf{Y}_i$. During the first iteration, $\mathbf{Q}$ and   $\mathbf{B}$ are null matrices and therefore we should drop off the last items in Step 7 and Step 9.
The equivalence of the  both algorithms is guaranteed with Theorem 1.
%, whose proof is given in the appendix. 
        \begin{theorem}\label{thm1_block}
           The $\mathbf{Q}$ and $\mathbf{B}$ obtained with Algorithm 3 satisfy: $\mathbf{Q}$ is orthonormal and  $\mathbf{B}=\mathbf{Q}^{\top}\mathbf{A}$.
        \end{theorem}
    \begin{proof}
                We prove Theorem 1 via induction. For any variable $v$ \emph{after} the $i$-th iteration of the loop is executed, we use $v^{(i)}$ to denote its value. Moreover, we assume the random matrix $\mathbf{\Omega}$ is of full column rank. In the base case, $\mathbf{Q}^{(1)}=\mathbf{Q}_1$ is orthonormal because of Step 8 in Algorithm 3. It also ensures that $\mathbf{Q}_i$ is orthonormal, and $\mathbf{Q}_i\mathbf{R}_i = \mathbf{Y}_i$. 
So,
\begin{equation}\label{eq:b1}
\begin{aligned}
\mathbf{B}^{(1)} = \mathbf{B}_1= \mathbf{R}_1^{-\top}\mathbf{\Omega}_1^{\top} \mathbf{A}^{\top}\mathbf{A} =(\mathbf{A} \mathbf{\Omega}_1 \mathbf{R}_1^{-1})^{\top}\mathbf{A}\\
 = (\mathbf{Y}_1\mathbf{R}_1^{-1})^{\top}\mathbf{A} = \left(\mathbf{Q}^{(1)}\right)^{\top}\mathbf{A} ~ .
\end{aligned}
\end{equation}

            \par Now, suppose the proposition holds for the $i$-th iteration. We need to prove $\mathbf{Q}^{(i+1)}$ is orthonormal and $\mathbf{B}^{(i+1)} = \left(\mathbf{Q}^{(i+1)}\right)^{\top}\mathbf{A}$. We first check the orthogonality of $\mathbf{Q}_{i+1}$.
%   \[   
         \begin{equation} \label{eq_proof_b1}
            \begin{aligned}
                \mathbf{Q}_{i+1}^{\top} \mathbf{Q}^{(i)} 
                &= \left(\left(\mathbf{A} - \mathbf{Q}^{(i)}\mathbf{B}^{(i)}\right)\mathbf{\Omega}_{i+1}\mathbf{R}_{i+1}^{-1}\right)^{\top}  \mathbf{Q}^{(i)}\\
                &= \left(\left(\mathbf{A} - \mathbf{Q}^{(i)}\left(\mathbf{Q}^{(i)}\right)^{\top}\mathbf{A}\right)\mathbf{\Omega}_{i+1}\mathbf{R}_{i+1}^{-1}\right)^{\top}  \mathbf{Q}^{(i)}\\
                &= \left((\mathbf{I} - \mathbf{P}_{Q^{(i)}})\mathbf{A}\mathbf{\Omega}_{i+1}\mathbf{R}_{i+1}^{-1}\right)^{\top}  \mathbf{Q}^{(i)}\\
%                &= \scalebox{0.95}{$\left(\mathbf{A}\mathbf{\Omega}_{i+1}\mathbf{R}_{i+1}^{-1}\right)^{\top}  (\mathbf{I} - \mathbf{P}_{Q^{(i)}})^{\top} \mathbf{Q}^{(i)}$}\\
 %  \]%        \end{equation}
                &= \left(\mathbf{A}\mathbf{\Omega}_{i+1}\mathbf{R}_{i+1}^{-1}\right)^{\top}  (\mathbf{Q}^{(i)}-\mathbf{P}_{Q^{(i)}}\mathbf{Q}^{(i)}) \\
                & = \mathbf{O}.
            \end{aligned}
           \end{equation}
The last two equalities of (\ref{eq_proof_b1}) is based on the properties of projector matrix $\mathbf{P}_{Q^{(i)}}\equiv \mathbf{Q}^{(i)}\left(\mathbf{Q}^{(i)}\right)^{\top}$.
See the Appendix.
% Note that $\mathbf{P}_{Q^{(i)}}$ is an orthogonal projector. 
Eq. (\ref{eq_proof_b1}) guarantees that $\mathbf{Q}^{(i+1)}$ is an orthonormal matrix. Then, 
%\begin{equation}
\[
            \begin{aligned}\label{eq_proof_b2}
                \mathbf{B}_{i+1} 
                &= \mathbf{R}_{i+1}^{-{\top}} \mathbf{\Omega}_{i+1}^{\top} \left( \mathbf{A}^{\top}\mathbf{A} - {\mathbf{B}^{(i)}}^{\top}{\mathbf{B}^{(i)}}\right) \\
                &\stackrel{1}{=} \mathbf{R}_{i+1}^{-{\top}} \mathbf{\Omega}_{i+1}^{\top} {\left(\mathbf{A}-\mathbf{Q}^{(i)}\mathbf{B}^{(i)}\right)}^{\top}\left(\mathbf{A}-\mathbf{Q}^{(i)}\mathbf{B}^{(i)}\right) \\
                &=   \left(\mathbf{A}\mathbf{\Omega}_{i+1}\mathbf{R}_{i+1}^{-1} -\mathbf{Q}^{(i)}\mathbf{B}^{(i)}\mathbf{\Omega}_{i+1}\mathbf{R}_{i+1}^{-1}\right)^{\top}\left(\mathbf{A}-\mathbf{Q}^{(i)}\mathbf{B}^{(i)}\right) \\
%&= \scalebox{0.95}{$\left(\mathbf{Y}_{i+1}\mathbf{R}_{i+1}^{-1}\right)^{\top}\left(\mathbf{A}-\mathbf{Q}^{(i)}\mathbf{B}^{(i)}\right)$} \\
&= \mathbf{Q}_{i+1}^{\top}\left(\mathbf{A}-\mathbf{Q}^{(i)}\mathbf{B}^{(i)}\right)  
                \stackrel{2}{=} \mathbf{Q}_{i+1}^{\top}\mathbf{A} ~,
            \end{aligned}
\]
%                        \end{equation}
where equality 1 holds due to $\mathbf{Q}^{(i)}$ is orthonormal and $\mathbf{B}^{(i)}=\left({\mathbf{Q}^{(i)}}\right)^{\top}  \mathbf{A}$. Equality 2 just follows from (\ref{eq_proof_b1}). 

Therefore, $\mathbf{Q}$ is orthonormal and $\mathbf{B} = \mathbf{Q}^{\top}\mathbf{A}$, based on  the induction hypothesis and Step 10 in Algorithm 3.
This ends the proof. 
\end{proof}

As the blocked randQB algorithm is mathematically equivalent to  the basic \texttt{randQB}, Algorithm 3 inherits the theoretical error bound (if ignoring the round-off error):
\begin{equation}\label{eq:statisticError}
\mathbb{E}\left( \|\mathbf{A-QB}\|_{\mathrm{F}} \right) \le \left( 1+ \frac{k}{s-1}  \right)^{1/2} \left( \sum_{j=k+1}^{\min(m,n)}\sigma_{jj}^2 \right)^{1/2},
\end{equation}
where $\mathbb{E}$ denotes expectation. We see that the theoretically minimal error is only magnified by a factor of $(1+\frac{k}{s-1})^{1/2}$. If measuring the error with $l_2$-norm, we have a similar error bound formula (see Theorem 10.6 of \cite{Halko2011review}). Moreover, it can be shown that the likelihood of a substantial deviation from the expectation is extremely small; see Sec. 10.3 of \cite{Halko2011review} for a proof.

\subsection{The Version with Re-Orthogonalization}
Due to the accumulation of round-off errors, the orthonormality among the columns in $\{\mathbf{Q}_1,\mathbf{Q}_2, \cdots \}$ may lose. This affects the correctness of some statements in Algorithm 3, and increases the error of its output. To fix this problem, we explicitly reproject $\mathbf{Q}_i$ away from the span of the previously computed basis vectors, just as what is done in \cite{Martin2016}. Then, the formula for matrix $\mathbf{B}_i$ is revised to incorporate the modified $\mathbf{Q}_i$.

The re-orthogonalization step corresponds to:
\begin{equation}\label{eq:re-orth1}
\tilde{\mathbf{Q}}_i\tilde{\mathbf{R}}_i=\mathbf{Q}_i-\mathbf{Q}\mathbf{Q}^{\top}\mathbf{Q}_i,
\end{equation}
where $\tilde{\mathbf{Q}}_i \neq \mathbf{Q}_i$ and $\tilde{\mathbf{R}}_i \neq \mathbf{I}$ due to round-off error. And, $\tilde{\mathbf{Q}}_i$ is better orthogonal to the previously generated $\{\mathbf{Q}_1, \mathbf{Q}_2, \cdots, \mathbf{Q}_{i-1}\}$ than $\mathbf{Q}_i$.
Since $\mathbf{Q}_i\mathbf{R}_i=\mathbf{Y}_i$, 
\begin{equation}\label{eq:tildeQi}
\tilde{\mathbf{Q}}_i= \left(\mathbf{I}-\mathbf{Q}\mathbf{Q}^{\top}\right) \mathbf{Y}_i \mathbf{R}_i^{-1} \tilde{\mathbf{R}}_i^{-1},
\end{equation}
%Then, the new $\boldsymbol{B}_i$ is calculated as: 
% \! is used to adjust space before/after =, +, -.
\begin{equation}\label{eq:newBi}
            \begin{aligned}
\tilde{\mathbf{B}}_i =& \tilde{\mathbf{Q}}_i^{\top} \mathbf{A}
=(\tilde{\mathbf{R}}_i\mathbf{R}_i)^{-{\top}} \mathbf{Y}_i^{\top}\left(\mathbf{I}-\mathbf{Q}\mathbf{Q}^{\top}\right) \mathbf{A}  \\
=& (\tilde{\mathbf{R}_i}\mathbf{R}_i)^{-\top} (\mathbf{\Omega}_i^{\top} \mathbf{A}^{\top} - \mathbf{\Omega}_i^{\top} \mathbf{B}^{\top} \mathbf{Q}^{\top}) (\mathbf{A}- \mathbf{Q}\mathbf{Q}^{\top}\mathbf{A}) \\
%\approx& \scalebox{0.805}{$(\tilde{\mathbf{R}_i}\mathbf{R}_i)^{-\!T} \left(\mathbf{H}_i^{\top}\! -\! \mathbf{G}_i^{\top} \mathbf{Q}^{(i-\!1)} \mathbf{B}^{(i-\!1)} \!-\!\mathbf{\Omega}_i^{\top} \left(\mathbf{B}^{(i-\!1)}\right)^{\top} \mathbf{B}^{(i-\!1)}\!+\! \mathbf{\Omega}_i^{\top} \left(\mathbf{B}^{(i-\!1)}\right)^{\top} \left(\mathbf{Q}^{(i-\!1)}\right)^{\top} \mathbf{Q}^{(i-\!1)} \mathbf{B}^{(i-\!1)} \right)$} \\
%\approx& \scalebox{0.95}{$(\tilde{\mathbf{R}}_i\mathbf{R}_i)^{-{\top}} \left(\mathbf{H}_i^{\top}\! -\! \mathbf{Y}_i^{\top} \mathbf{Q}\mathbf{B}\! -\! \mathbf{\Omega}_i^{\top} \mathbf{B}^{\top} \mathbf{B} \right)$},
=& (\tilde{\mathbf{R}}_i\mathbf{R}_i)^{-{\top}} \left(\mathbf{H}_i^{\top}\! -\! \mathbf{Y}_i^{\top} \mathbf{Q}\mathbf{B}\! -\! \mathbf{\Omega}_i^{\top} \mathbf{B}^{\top} \mathbf{B} \right),
            \end{aligned}
\end{equation}
where $\mathbf{H}_i$ denotes $\mathbf{H}(:, (i-1)b+1:ib)$. 
The last equality utilizes that $\mathbf{B}=\mathbf{Q}^{\top}\mathbf{A}$,
although this may not hold after a large number of iterations due to numerical round-off error.
%In the deduction, the orthogonal property of $\mathbf{Q}$  is used as less as possible, since it may not hold in the floating-point computation. 

Based on (\ref{eq:re-orth1}) and (\ref{eq:newBi}), the version with re-orthogonalization can be obtained  by replacing Step 9 in Algorithm 3 with the following steps:

%\noindent{\begin{tabular}{l}
{\centering
\begin{tabular}{l}
\toprule
9: ~$[\mathbf{Q}_i,~\tilde{\mathbf{R}}_i]= $ qr$(\mathbf{Q}_i-\mathbf{Q}(\mathbf{Q}^{\top}\mathbf{Q}_i))$; \\
9': $\mathbf{R}_i=\tilde{\mathbf{R}}_i\mathbf{R}_i$; \\
9'': $\mathbf{B}_i\!=\! \mathbf{R}_i^{-{\top}}(\mathbf{H}(:,~ (i\!-\!1)b\!+\!1:ib)^{\top}\!-\! \mathbf{Y}_i^{\top}\mathbf{Q}\mathbf{B} \!-\! \mathbf{\Omega}_i^{\top}\mathbf{B}^{\top}\mathbf{B})$; \\
\bottomrule
\end{tabular}
}

\noindent{Here, $\mathbf{Q}_i$ and $\mathbf{B}_i$ are overwritten to stand for $\tilde{\mathbf{Q}}_i$ and $\tilde{\mathbf{B}}_i$.}

\subsection{The Single-Pass Algorithm for PCA}
An important feature of Algorithm 3 is that Steps 3 and 4 can be executed with only one pass over matrix $\mathbf{A}$. Suppose $\mathbf{a}_i$ and $\mathbf{g}_i$ denote the $i$-th \emph{rows} of matrix $\mathbf{A}$ and $\mathbf{G}$, respectively. 
\begin{equation}
\mathbf{H}= [ \mathbf{a}_1^{\top}, \mathbf{a}_2^{\top}, \cdots, \mathbf{a}_m^{\top} ] \\
\left[ \begin{array}{c}
                 \mathbf{g}_1\\
                 \mathbf{g}_2\\
                 \vdots\\
                 \mathbf{g}_m\\
            \end{array}\right] \\
  =\sum_{i=1}^m \mathbf{a}_i^{\top}\mathbf{g}_i        .
\end{equation}
So, with the $i$-th row of $\mathbf{A}$, we can calculate the $i$-th row of $\mathbf{G}$ with Step 3, and then the $i$-th item in the summation for calculating $\mathbf{H}$ as (10). Combined with the over-sampling, the single-pass algorithm for computing PCA is as Algorithm 4.
 
        \begin{algorithm}
        \caption{A single-pass algorithm for computing PCA}
        \label{alg3}
        \begin{algorithmic}[1]
            \REQUIRE $\mathbf{A}\in \mathbb{R}^{m\times n}$, rank parameter $k$, block size $b$.
            \STATE $\mathbf{Q} = [~]; ~ \mathbf{B} = [~];$
            \STATE Choose $l=tb$, which is slightly larger than $k$;
            \STATE $\mathbf{\Omega}= $ randn($n, l$); $\mathbf{G} = [~]$; Set $\mathbf{H}$ to an $n\times l$ zero matrix;
            \WHILE{$\mathbf{A}$ is not completely read through}
            \STATE Read next few rows of $\mathbf{A}$ into RAM, denoted by $\mathbf{a}$;
            \STATE $\mathbf{g}=\mathbf{a} \mathbf{\Omega}$; ~ $\mathbf{G}= [\mathbf{G}; ~ \mathbf{g}]$;
            \STATE $\mathbf{H}=\mathbf{H}+\mathbf{a}^\top \mathbf{g}$;
            \ENDWHILE
     %   \STATE Compute $\mathbf{G} = \mathbf{A\Omega}$ and $\mathbf{H} = \mathbf{A}^{\top}\mathbf{G}$ in a single pass over $\mathbf{A}$;
            \FOR{$i=1, 2, \cdots, t$}
            %\FOR{$\sqrt{E} > \varepsilon$}
            \STATE $\mathbf{\Omega}_i = \mathbf{\Omega}(:,~ (i-1)b+1:ib)$;
            \STATE $\mathbf{Y}_i = \mathbf{G}(:,~ (i-1)b+1:ib) - \mathbf{Q}(\mathbf{B\Omega}_i)$;
			\STATE $[\mathbf{Q}_i,~\mathbf{R}_i]= $ qr($\mathbf{Y}_i$);
%			\STATE $\mathbf{B}_i= \mathbf{R}_i^{-T}(\mathbf{H}(:,~ (i-1)b+1:ib)^{\top}- \mathbf{\Omega}_i^{\top}\mathbf{B}^{\top}\mathbf{B}) $;
\STATE $[\mathbf{Q}_i,~\tilde{\mathbf{R}}_i]= $ qr$(\mathbf{Q}_i-\mathbf{Q}(\mathbf{Q}^{\top}\mathbf{Q}_i))$;
\STATE $\mathbf{R}_i=\tilde{\mathbf{R}}_i\mathbf{R}_i$;
\STATE $\mathbf{B}_i\!=\! \mathbf{R}_i^{-{\top}}(\mathbf{H}(:,~ (i\!-\!1)b\!+\!1:ib)^{\top}\!-\! \mathbf{Y}_i^{\top}\mathbf{Q}\mathbf{B} \!-\! \mathbf{\Omega}_i^{\top}\mathbf{B}^{\top}\mathbf{B})$;                         
				\STATE $\mathbf{Q} = [\mathbf{Q}, ~\mathbf{Q}_i]$;~ $\mathbf{B} = [\mathbf{B}^{\top},  ~\mathbf{B}_i^{\top}]^{\top}$;
%            \STATE $\mathbf{B} = \left[ \begin{array}{c}
%                 \mathbf{B}\\
%                 \mathbf{B}_i
%            \end{array}\right]$;
            \ENDFOR
            \STATE $[\tilde{\mathbf{U}}, \mathbf{S}, \mathbf{V}]$= svd($\mathbf{B}$);
        	  \STATE $\mathbf{U}= \mathbf{Q}\tilde{\mathbf{U}}$;
        	  \STATE $\mathbf{U}=\mathbf{U}(:, 1:k)$; $\mathbf{V}=\mathbf{V}(:, 1:k)$; $\mathbf{S}= \mathbf{S}(1:k,1:k)$;
        	  \STATE \textbf{return} $\mathbf{U, S, V}$.
        \end{algorithmic}
        \end{algorithm}

In the algorithm, the while loop corresponds to Steps 3 and 4 in Algorithm 3, but involves only one pass over $\mathbf{A}$. 
In every step, small matrices in size $m\times l$ or $n\times l$ (noting $l \ll \min(m, n)$ ) are used. So, the memory cost of this algorithm is small, which can be bounded by that for storing $(m+2n)l$ floating numbers. The computational cost of this algorithm is the same as Algorithm 3 and 1 \cite{Halko2011}, i.e., $\mathcal{O}(mnk)$ or $\mathcal{O}(mn\log(k))$ flops. The theoretical error bounds of $\mathbf{A-QB}$ also apply to $\mathbf{A-USV^\top}$ in Algorithm 4, as the latter hardly induces new error.

This single-pass algorithm requests that the data matrix $\mathbf{A}$ is stored in a row-major format. If it is given in a column-major format,  we can apply the algorithm to $\mathbf{A}^{\top}$ instead.

  In case there is a request for higher accuracy, the power scheme can be applied with a small $P$. If $P\!=\!1$, it is equivalent to replacing $\mathbf{A}$ with $\mathbf{A}\mathbf{A}^{\top}\mathbf{A}$ in the algorithm. It can be implemented by adding one pass over  $\mathbf{A}$, even the orthonormalization is enforced for better accuracy  \cite{Martin2016,RSVDPACK}. Nevertheless, the single-pass algorithm works well in many applications.

\section{Experiments}

All experiments are carried out on a Linux server with two 12-core Intel Xeon E5-2630 CPUs (2.30 GHz), 15 MB of L3 cache, and 32 GB RAM. The algorithms have been implemented in C with OpenMP derivatives for multi-thread computing. The compiler used is Intel ICC with MKL libraries \cite{Intel}. The QR factorization and other basic linear algebra operations are implemented through LAPACK routines which are automatically executed in parallel on the multi-core CPUs.

We first validate the accuracy of the proposed single-pass algorithm. Then, large test cases stored on hard disk in IEEE single-precision float format are used to validate the algorithm's efficiency. In all experiments, the block size $b=10$.

\subsection{Accuracy Validation}
We consider test matrices owning the following singular spectrums with different decaying behavior, where $\sigma_{ii}$ denotes the $i$-th singular value (i.e., a diagonal element of matrix $\mathbf{\Sigma}$). 
    \begin{itemize}
        \item Type 1: 
        \[ 
\setlength{\abovedisplayskip}{3pt}
\setlength{\belowdisplayskip}{3pt}
        \sigma_{ii} = \begin{cases}10^{-4(i-1)/19}, ~~~~~~~~~~i= 1, 2, \cdots, 20, \\
        10^{-4}/(i-20)^{1/10}, ~i=21, 22, \cdots, \min(m,n). \end{cases}        \]
        \item Type 2: $\sigma_{ii} = i^{-2},~~ i= 1, 2, \cdots$.
        \item Type 3: $\sigma_{ii} = i^{-3},~~ i= 1, 2, \cdots$.
        \item Type 4: $\sigma_{ii} = e^{-i/7},~~ i= 1, 2, \cdots$.
        \item Type 5: $\sigma_{ii} = 10^{-i/10},~~ i= 1, 2, \cdots$.
    \end{itemize}    
Type 1 is from \cite{Halko2011}, and Type 3 and Type 5 are from \cite{Mary2015}. 
These singular value distributions are shown in Fig. 2. It reveals that the singular values of Type 1 and Type 2 matrices decay asymptotically slowly, although they attenuate very fast at the start. The singular values of Type 4 and Type 5 matrices decay asymptotically faster.
\begin{figure}[h]
\centering
\setlength{\abovecaptionskip}{0.02 cm}
\setlength{\fboxsep}{0.5pt}
\includegraphics[height=2in]{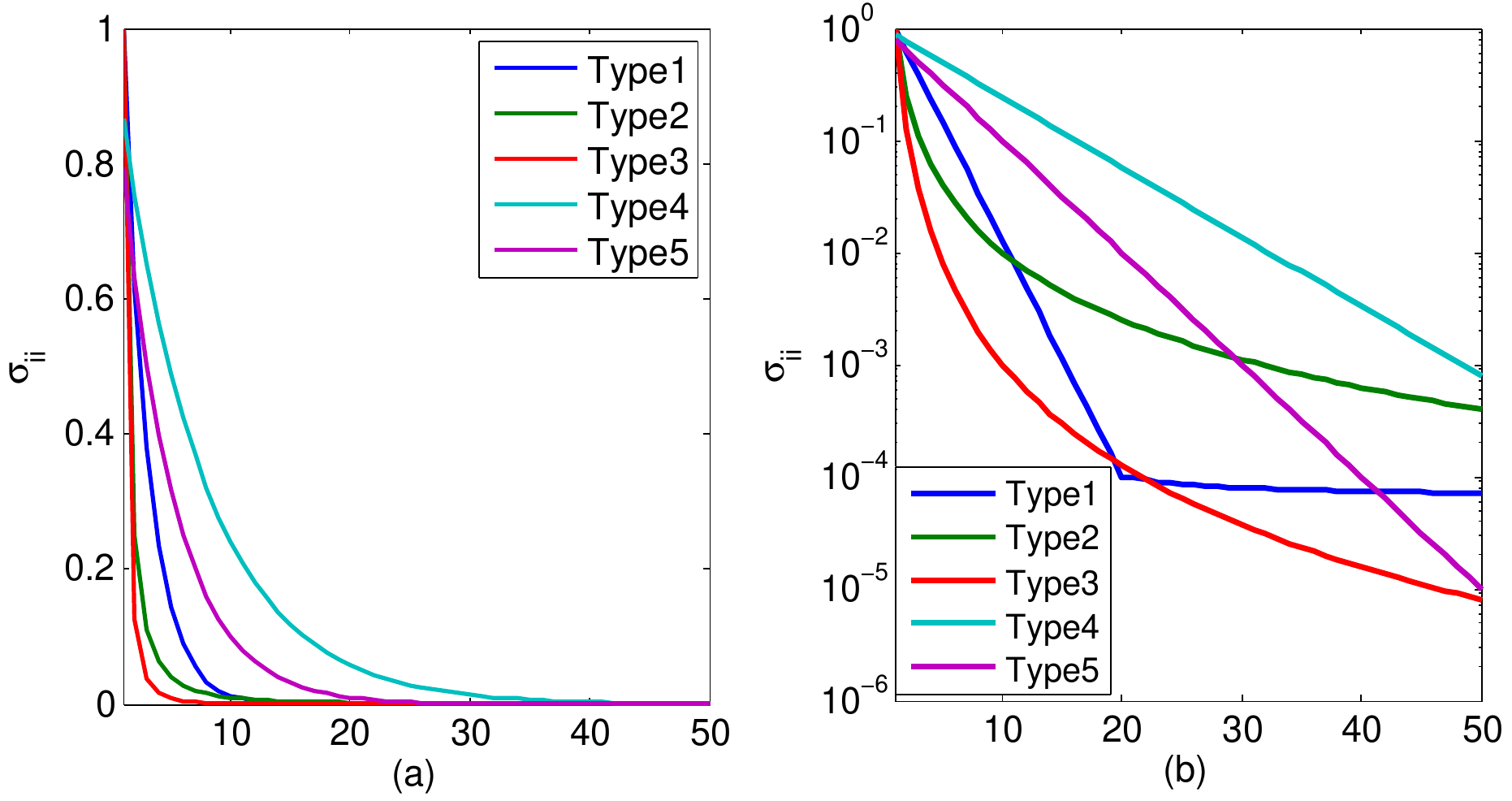}
\caption{Different decay behavior of the singular values of the test matrices. (a) Normal plot, (b) Semi-logarithmic plot.}
\label{fig2}
\end{figure} 
 \begin{figure}[h]         \centering
\setlength{\abovecaptionskip}{0.02 cm}
\subfigure[Type 2 matrix] {\includegraphics[height=2in]{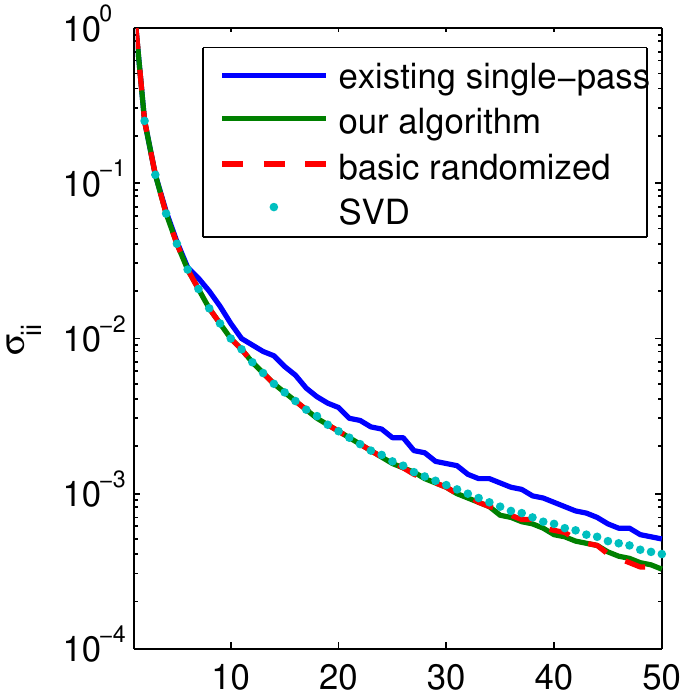}}
\subfigure[Type 4 matrix] {\includegraphics[height=2in]{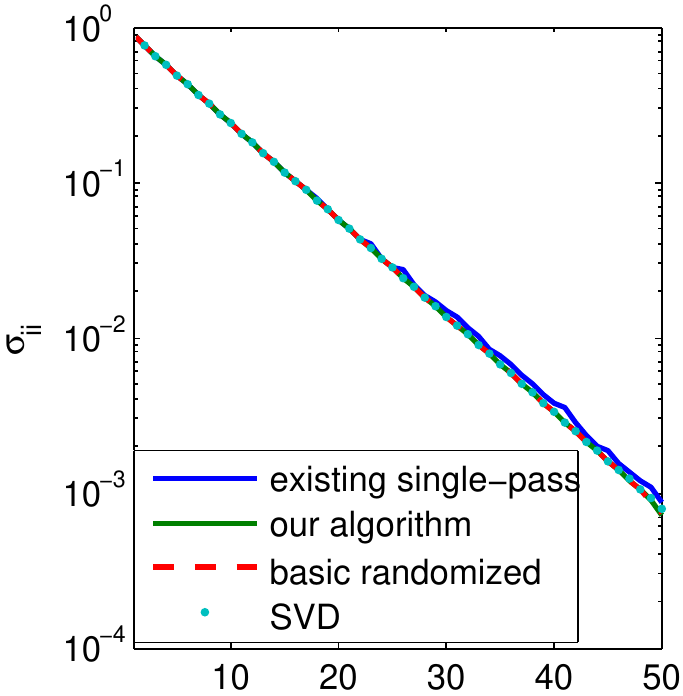}}
        \caption{The computed singular values for a slow-decay and a fast-decay matrix, showing the accuracy of our algorithm.}
        \label {Fig3}
        \end{figure}

 For each type, we construct a $3000\times 3000$ matrix through multiplying $\mathbf{\Sigma}$ with randomly drawn orthogonal 
matrices $\mathbf{U}$ and $\mathbf{V}$. We compute the first 50 singular values and singular vectors for each matrix with the basic randomized Algorithm 1, the existing single-pass algorithm (Algorithm 2) and our Algorithm 4, and compare the results with the accurate values obtained by SVD. 
   The over-sampling parameter is set to 10 (i.e., $l=60$). Fig. 3 shows the computed singular values of two matrices, which demonstrates the single-pass algorithm in \cite{Halko2011review} produces much larger error, and the results of  Algorithms 1 and 4 are indistinguishable. It also reveals that the algorithms produce better results for matrices with asymptotically faster decay of singular values. This is a common property of the randomized algorithms based on QB approximation \cite{Halko2011review,Mary2015,Martin2016}. So, we will focus on the accuracy for the matrices with slow decay of singular values.
\begin{figure}[h]
\centering 
\setlength{\abovecaptionskip}{0.03 cm}
\setlength{\fboxsep}{0.5pt}
\includegraphics[height=2in]{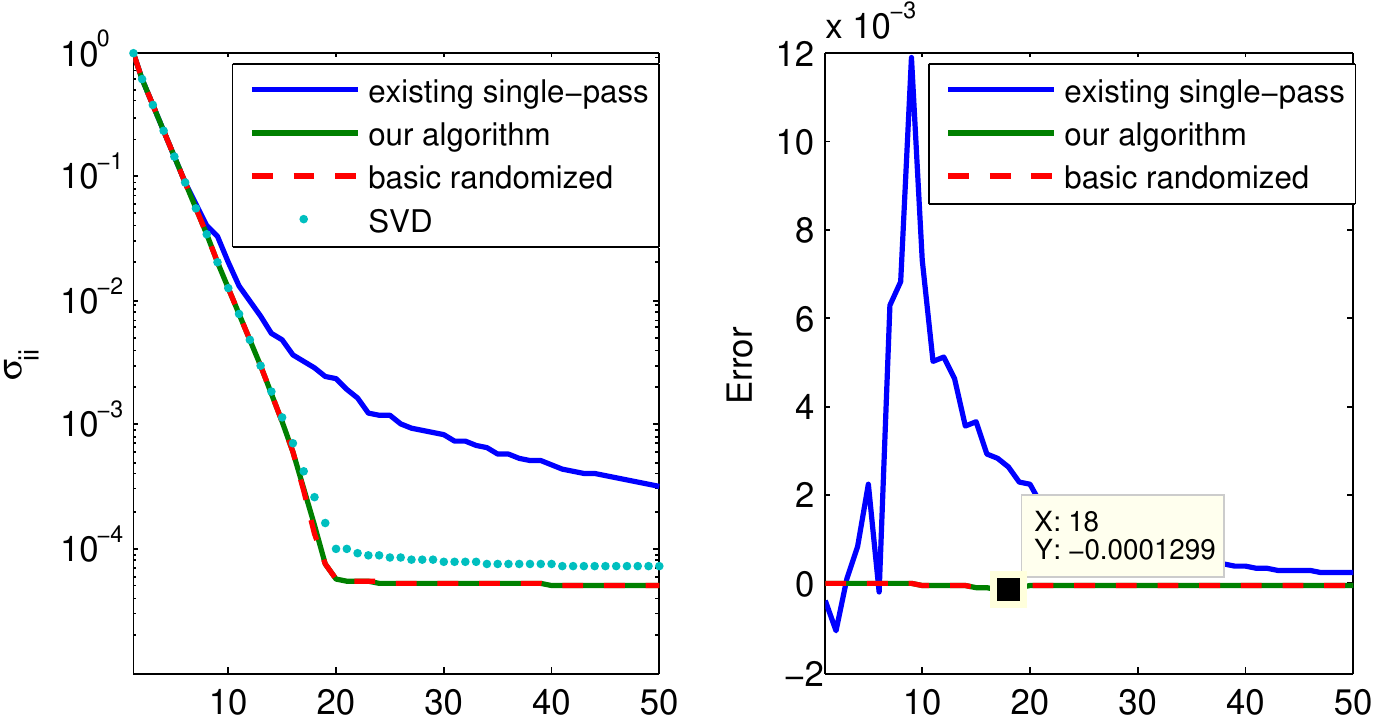}
\caption{The computed singular values and their absolute errors for a very slow-decay matrix (Type 1), showing the advantage of our algorithm over Algorithm 2.}
\label{fig4}
\end{figure} 

For the Type 1 matrix, the accuracy of the randomized algorithms all decreases; the existing single-pass algorithm \cite{Halko2011review} produces considerably large error (up to $1.2\!\times\! 10^{-2}$), as shown in Fig. 4. 
 While using the proposed Algorithm 4, we can reduce the maximum error to $1.3\!\times\! 10^{-4}$ ($\sim$ \textbf{92X} \textbf{smaller}). And, its accuracy looks acceptable. % This validates the accuracy advantage of our algorithm over the existing single-pass algorithm.
Fig. 5(a) shows the first principal components (i.e., $\mathbf{v}_1$) computed by SVD and our algorithm respectively, which looks indistinguishable. Their difference in  $l_\infty$-norm is only $2.8\!\times\! 10^{-5}$. For the other principal components, we calculate the correlation coefficient between the results obtained with the both methods individually. As shown in Fig. 5(b), the correlation coefficients are close to 1 (meaning an exact equality of two vectors). The largest difference occurs at the 10th principal component, with a correlation coefficient of 0.9993.  
For other matrices with faster decay of singular
 values (Types 2$\sim$5), the randomized algorithm exhibits better accuracy and outputs more accurate principal components.
\begin{figure}[h]
\centering 
\includegraphics[height=2in]{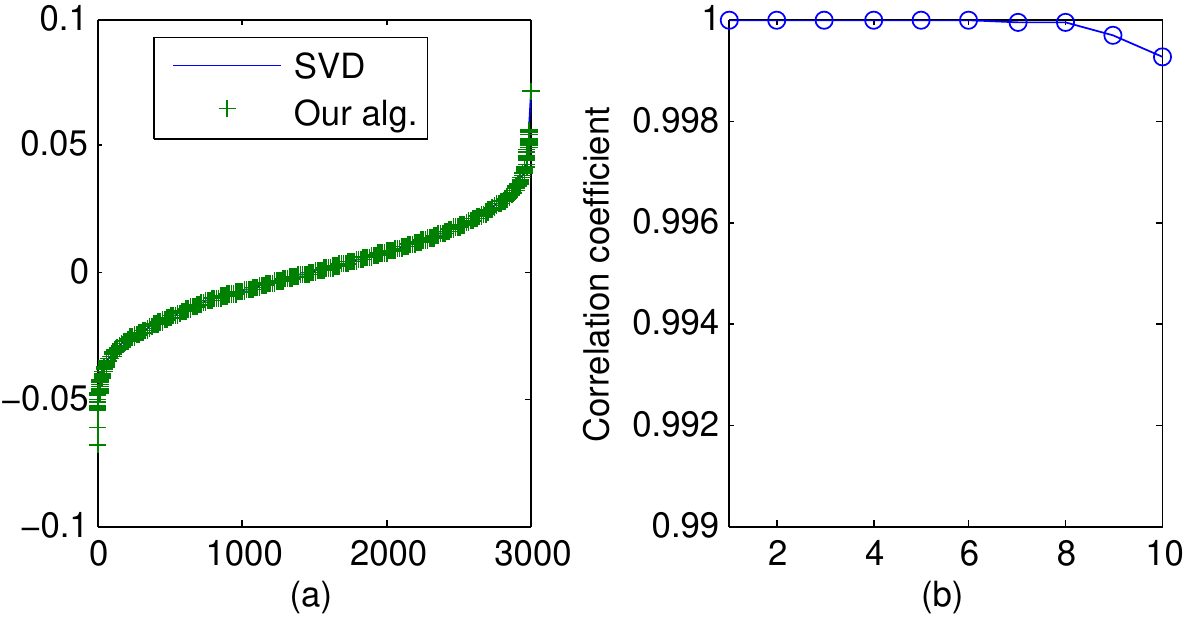}
\caption{The accuracy of our algorithm on principal components (with comparison to the results from SVD). (a) The numeric values of the first principal component ($\mathbf{v}_1$). (b) The correlation coefficients for the first 10 principal components.}
\label{fig5}
\end{figure}  

\subsection{Runtime Comparison}

Following \cite{Halko2011}, we construct several large data using the unitary discrete cosine transform (command ``\texttt{dct}'' in Matlab). They are $200,000\!\times\! 200,000$ matrices following the singular value distributions given in last subsection. Each matrix is stored as a 149 GB file on hard disk. 
We use \emph{fread} function to read the file and run Algorithm 4 for computing PCA. Each time we read $l$ rows of matrix, to avoid extra memory cost. Once they are loaded into RAM, the data are converted to the IEEE double-precision format. 
Algorithm 1 and Algorithm 2 are also tested for comparison.

Some results for the matrices with slow-decay singular values are listed in Table 1. $l$ in Algorithm 4 is set to 20 or 30. $t_{read}$ and $t_{PCA}$ mean the total time (in seconds) for reading the data and the total runtime of the algorithm (including $t_{read}$), respectively. ``max\_err'' is the maximum error of the computed singular values. From the table we see that the time for reading data dominates the total runtime, and the proposed algorithm is about \textbf{2X faster} than the basic randomized algorithm used in \cite{Halko2011} while keeping same accuracy. % The variation of $t_{read}$ may be due to the fragments of disk space.
If comparing Algorithm 2 and ours, we see that the former may be slightly faster but produces much larger error.

\begin{table}[h]
  \caption{The results for several $200,000\!\times\! 200,000$ data, which demonstrate the efficiency of our Algorithm 4 (time in unit of second).}
  \label{tab:table1}
  \centering
%\small{
%\begin{spacing}{0.9}
\renewcommand{\multirowsetup}{\centering}
%  \begin{tabular}{c c c c c c c c c c c} 
% matrix, m, n, k, t_read, t_PCA, max_err
 \begin{tabular}{@{~}c@{~~}c@{~~}c@{~~}c@{~~}c@{~~}c@{~~}c@{~~}c@{~~}c@{~~}c@{~~}c@{~~}} 
  \toprule
  \multirow{2}{*}{Matrix}  & \multirow{2}{*}{$k$} &
  \multicolumn{3}{c}{Algorithm 1} & \multicolumn{3}{c}{Algorithm 2} &\multicolumn{3}{c}{Algorithm 4} \\
 \cmidrule(r){3-5} \cmidrule(r){6-8} \cmidrule(r){9-11}
 &  &  $t_{read}$ & $t_{PCA}$ & max\_err & $t_{read}$ & $t_{PCA}$ & max\_err &  $t_{read}$ & $t_{PCA}$ & max\_err \\
\midrule
Type1 & 16 & 2390 & 2607 & 1.7e-3 & 1186 & 1404 & 2.2e-2 & 1206 & 1426 & 1.8e-3 \\
Type1  & 20 & 2420 & 2616 & 9e-4 & 1198   & 1380  & 1.6e-1  & 1217 & 1413 & 1.2e-3  \\ 
  %\cmidrule(r){3-8}
Type1 & 24 & 2401 & 2593 & 1e-3 & 1216  & 1400  & 1.5e-1  & 1216 & 1414 & 1.2e-3  \\ 
Type2 & 12 & 2553 & 2764 & 5e-4 & 1267 & 1477 & 3e-2 & 1276 & 1490 & 5e-4  \\
Type3 & 24 & 2587 & 2777 & 1e-5 & 1312 & 1500 & 1.7e-3 & 1310 & 1502 & 2e-5 \\ 
% & 200,000 & 20,000 & 12 & xxx & xxx & 4.5 & 27.8 & xxx  \\ 
% & 500,000 & 80,000 & 12 & xxx & xxx & 1074 & 1477 & xxx  \\ 
\bottomrule 
 \end{tabular}
% \end{spacing}
% }
\end{table}

To improve the accuracy, the power scheme with $P\!=\!1$ could be applied, which corresponds to one more pass over the data. For Algorithm 4, we just run the while loop once again with $\mathbf{\Omega}$ replaced by $\mathbf{H}$ after ``orth'' operation.  
%
%\noindent{Computer $\mathbf{G\!=\!AH}$ and $\mathbf{H\!=\!A}^{\top}\mathbf{G}$ in a single pass over $\mathbf{A}$; }
In our experiments, this two-pass algorithm has similar runtime as Algorithm 1, but dramatically reduces ``max\_err'' to $4.6\times 10^{-7}$ and $3\times 10^{-6}$ for the Type 1 and Type 2 matrices, respectively.  

In these experiments, the memory cost of Algorithm 4 ranges from \textbf{402 MB} to \textbf{490 MB}. In contrast, the standard SVD (including the ``\texttt{svds}'' in Matlab for truncated SVD) requests much larger memory than the available physical RAM, and therefore does not work. To take a taste of how fast the proposed randomized algorithm runs, we test a 10,000$\times$10,000 matrix. Performing a complete SVD and ``\texttt{svds}'' for the first 50 principal components take 226 and 219 seconds, respectively, while the proposed algorithm costs only 0.69 seconds.

\subsection{Real Data}

We apply the single-pass algorithm with $k\!=\!50$ to the matrix representing the images of faces from the FERET database \cite{Phillips2000feret}. As in \cite{Halko2011}, we add two duplicates for each image into the data. For each duplicate, the value of a random choice of 10\% of the pixels is set to random numbers uniformly chosen from $0, 1, \cdots, 255$. This forms a $102,042\!\times \! 393,216$ matrix, whose rows consist of images.
Before processing, we normalize the matrix by subtracting from each row its mean, and then dividing it by its Euclidean norm. With the proposed algorithm, it takes 1453 seconds ($\sim$ \textbf{24 minutes}) to process all 150 GB of this data stored on disk. The computed singular values are plotted in Fig. 6. We have also checked the computed ``eigenfaces'', which well match those presented in \cite{Halko2011}.

\begin{figure}[h]         \centering
\subfigure[Computed singular values] {\includegraphics[width=2.9in]{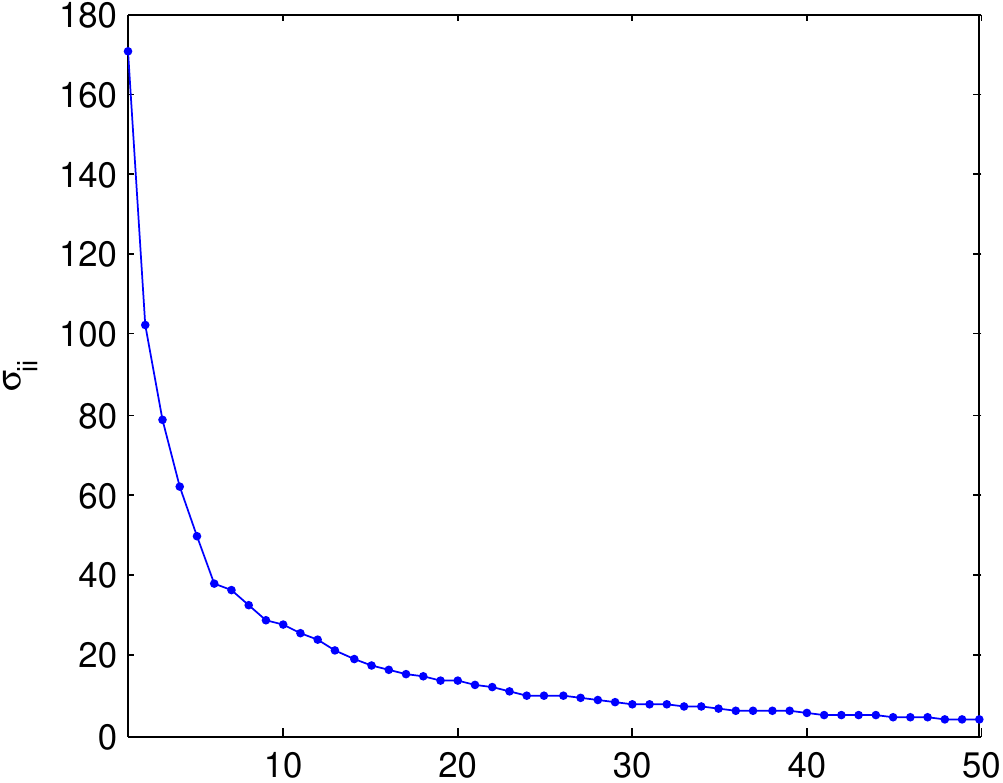}}
\subfigure[Four eigenfaces] {\includegraphics[width=1.53in]{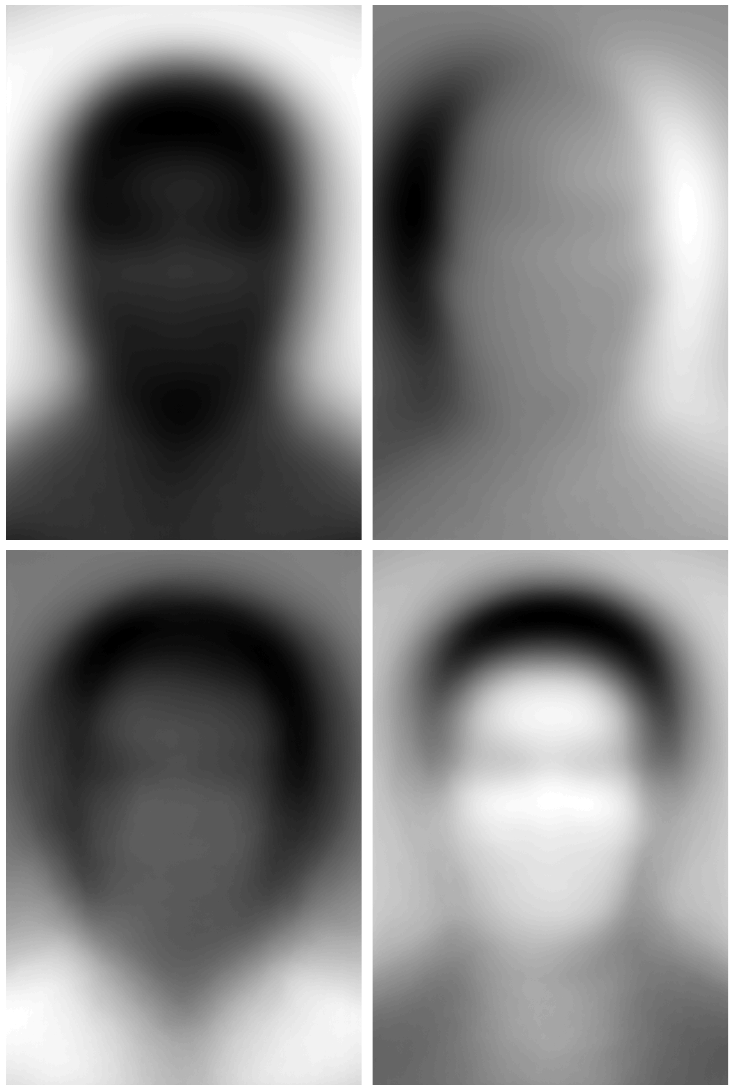}}
        \caption{The computational results for the FERET matrix.}
        \label {Fig3}
        \end{figure} 

\section{Conclusions} 
An algorithm for single-pass PCA of large and high-dimensional data is proposed. It involves only one pass over the data, and keeps the comparable accuracy to the existing randomized algorithms. Experiments demonstrate the algorithm's effectiveness for computing the principal components of large-size ($\sim$150 GB) data with high dimension that cannot be fit in memory, in terms of runtime and memory usage. 

%As a complementary, we also show that the accuracy of the single-pass algorithm can be largely improved by adding one more pass over the data. 
The improvement of the proposed algorithm by utilizing structured random matrix, and its applications to streaming data and on distributed computing environment will be explored in the future. 

%The codes of the proposed algorithm and related test data
%will be shared on the website of the first author.
%The algorithm runs efficiently through exploiting modern computer architecture, and is able to compute several tens of principal components and/or singular values accurately. Experiments on a computer with multi-core CPUs demonstrate the accuracy and efficiency of proposed algorithm, and the effectiveness of a remedy solution with one more pass over data for the situation where higher accuracy is requested.

\section{ Acknowledgments}
Portions of the research in this paper use the FERET database of facial
images collected under the FERET program, sponsored by the DOD Counterdrug
Technology Development Program Office.

\section{Appendix: Orthogonal Projector Matrix}
An orthogonal projector matrix corresponds to a linear transformation which converts any vector to its orthogonal projection on a subspace. The projector matrix is uniquely determined by the subspace. 
Let $\mathbf{P}_X$ denote the projector matrix corresponding to the orthogonal projection transformation onto $range(\mathbf{X})$. 
Based on the theory of linear least squares, if $\mathbf{X}$ has full column rank \cite{MatrixBook}, 
\begin{equation}\label{eq:projector}
\mathbf{P}_X = \mathbf{X}(\mathbf{X}^{\top}\mathbf{X})^{-1}\mathbf{X}^{\top} ~.
\end{equation}
It is simplified to $\mathbf{P}_X = \mathbf{X}\mathbf{X}^{\top}$, 
if $\mathbf{X}$ is an orthonormal matrix. 
%$\mathbf{P}_A$ has the following property:
%%\begin{equation}
%\begin{align}\label{eq:projector3}
%\forall \mathbf{x} \in \mathbb{R}^m, \mathbf{P}_A \mathbf{x}\in range(\mathbf{A}), \\
%\forall \mathbf{x} \in range(\mathbf{A}),~ \mathbf{P}_A \mathbf{x}= \mathbf{x}.
%\end{align}
%%\end{equation}
Obviously, $range(\mathbf{P}_A)=range(\mathbf{A})$. From (\ref{eq:projector}), it is easy to derive the following properties of a projector matrix.
%
%verify that the orthogonal projector is a symmetric matrix, and $\mathbf{P}_A^2 = \mathbf{P}_A$. 
%And, the orthogonal projector determined by the orthogonal complement of $range(\mathbf{A})$ is $\mathbf{I}-\mathbf{P}_A$.
%%, where $\mathbf{I}$ is the identity matrix \cite{MatrixBook}.
%There is a useful lemma.
    \begin{lemma}\label{perp_prod_zero}
        For a real-valued matrix $\mathbf{X}$ with full column rank,
    \begin{itemize}
    	\item $\mathbf{P}_X$ is a symmetric matrix.
    	\item $\mathbf{P}_X^2 = \mathbf{P}_X$.
    	\item $\mathbf{I}-\mathbf{P}_X$ is the orthogonal projector determined by the orthogonal complement of $range(\mathbf{X})$.
    	\item $\mathbf{P}_X \mathbf{X} - \mathbf{X} = \mathbf{O}$ , where $\mathbf{O}$ is the zero matrix.
    \end{itemize}    
    \end{lemma}


\begin{thebibliography}{5}
%
\bibitem {bahdanau2014neural}
Bahdanau, D., Cho, K., Bengio, Y.:
Neural machine translation by jointly learning to align and translate.
arXiv preprint arXiv:1409.0473 (2014)

\bibitem {kim2014convolutional}
Kim, Y.:
Convolutional neural networks for sentence classification.
arXiv preprint arXiv:1408.5882 (2014)

\bibitem {bengio2003neural}
Bengio, Y., Ducharme, R., Vincent, P., Jauvin, C.:
A neural probabilistic language model.
J. Mach. Learn. Res 3, 1137--1155 (2003)

\bibitem {mikolov2013distributed}
Mikolov, T., Sutskever, I., Chen, K., Corrado, G. S., Dean, J.:
Distributed representations of words and phrases and their compositionality.
Proc. NIPS'2013, 3111-3119 (2013)

\bibitem {Halko2011}
Halko, N., Martinsson, P.-G., Shkolnisky, Y., Tygert, M.:
An algorithm for the principal component analysis of large data sets.
SIAM J. Sci. Comput. 33, 2580--2594 (2011)

\bibitem {Learning2001}
Friedman, J., Hastie, T., Tibshirani, R.:
The Elements of Statistical Learning.
Springer series in statistics Springer, Berlin (2001)

\bibitem {Halko2011review}
Halko, N., Martinsson, P.-G., Tropp J. A.:
Finding structure with randomness: Probabilistic algorithms for constructing approximate matrix decompositions.
SIAM Review 53, 217--288 (2011)

\bibitem {ZhangZ2016}
Ye, Q., Luo, L., Zhang, Z.:
Frequent direction algorithms for approximate matrix multiplication with applications in \protect{CCA}.
Proc. IJCAI'16, 2301--2307 (2016)

\bibitem {TRIEST2016}
De Stefani, L., Epasto, A., Riondato, M., Upfal, E.:
{TRIEST}: Counting local and global triangles in fully-dynamic streams with fixed memory size.
Proc. SIGKDD'2016 (2016)

\bibitem {drineas2005nystrom}
Drineas, P., Mahoney, M. W.:
On the Nystr{\"o}m method for approximating a Gram matrix for improved kernel-based learning.
J. Mach. Learn. Res 6, 2153--2175 (2005)

\bibitem {wang2013improving}
Wang, S., Zhang, Z.:
Improving CUR matrix decomposition and the Nystr{\"o}m approximation via adaptive sampling.
J. Mach. Learn. Res 14, 2729--2769 (2013)

\bibitem {gittens2013revisiting}
Gittens, A., Mahoney, M. W.:
Revisiting the Nystr{\"o}m method for improved large-scale machine learning.
J. Mach. Learn. Res 17, 1--65 (2016)

%\bibitem {williams2000using}
%Williams, C. K., Seeger, M.:
%Using the Nystr{\"o}m method to speed up kernel machines.
%Proc. NIPS'2000, 661--667 (2000)

\bibitem {wang2016towards}
Wang, S., Zhang, Z., Zhang, T.:
Towards more efficient SPSD matrix approximation and CUR matrix decomposition.
J. Mach. Learn. Res 17, 1--49 (2016)

\bibitem {LargePCA2014}
Ordonez, C., Mohanam, N., Garcia-Alvarado, C.:
\protect{PCA} for large data sets with parallel data summarization.
Distrib. Parallel Databases 32, 377--403 (2014)

\bibitem {singlePCA_NIPS16}
Wu, S., Bhojanapalli, S., Sanghavi, S., Dimakis, A. G.:
Single pass {PCA} of matrix products.
Proc. NIPS'2016, 2585--2593 (2016)

\bibitem {Liberty2013}
Liberty, E.:
Simple and deterministic matrix sketching.
Proc. SIGKDD'2013, 581--588 (2013)

\bibitem {woodruff2014sketching}
Woodruff, D. P.:
Sketching as a tool for numerical linear algebra.
Foundations and Trends{\textregistered} in Theoretical Computer Science 10, 1--157 (2014)

\bibitem {mahoney2011randomized}
Mahoney, M. W.:
Randomized algorithms for matrices and data.
Foundations and Trends{\textregistered} in Machine Learning 3, 123--224 (2011)

\bibitem {CACM2016}
Drineas, P., Mahoney, M. W.:
\protect{RandNLA}: Randomized numerical linear algebra.
Communications of the ACM 59, 80--90 (2016)

\bibitem {ZhangW2016}
Zhang, W., Zhang, L., Jin, R., Cai, D., He, X.:
Accelerated sparse linear regression via random projection.
Proc. AAAI'16, 2337--2343 (2016)

\bibitem {wangss2015}
Wang, S.:
A practical guide to randomized matrix computations with {MATLAB} implementations.
arXiv preprint arXiv:1505.07570 (2015)

\bibitem {Woodruff2014}
Liang, Y., Balcan, M.-F. F., Kanchanapally, V., Woodruff, D.:
Improved distributed principal component analysis.
Proc. NIPS'2014, 3113--3121 (2014)

\bibitem {Mary2015}
Mary, T., Yamazaki, I., Kurzak, J., Luszczek, P., Tomov, S., Dongarra, J.:
Performance of random sampling for computing low-rank approximations of a dense matrix on {GPU}s.
Proc. SC'15, 60:1--60:11 (2015)

\bibitem {Martin2016}
Martinsson, P.-G., Voronin, S.:
A randomized blocked algorithm for efficiently computing rank-revealing factorizations of matrices.
SIAM J. Sci. Comput. 38, S485--S507 (2016)

\bibitem {EY1936}
Eckart, C., Young, G.:
The approximation of one matrix by another of lower rank.
Psychometrika 1, 211--218 (1936)

\bibitem {MatrixBook}
Golub, G. H., Van Loan, C. F.:
Matrix Computations.
Johns Hopkins University Press (1996)

\bibitem {Rokhlin2009}
Rokhlin, V., Szlam, A., Tygert, M.:
A randomized algorithm for principal component analysis.
SIAM J. Matrix. Anal. Appl. 31, 1100--1124 (2009)

\bibitem {Candes2015Algori}
Rafi, W., Candes, E.:
Randomized algorithms for low-rank matrix factorizations: Sharp performance bounds.
Algorithmica 72, 264--281 (2015)

\bibitem {RSVDPACK}
Voronin, S., Martinsson, P.-G.:
{RSVDPACK}: Subroutines for computing partial singular value
decompositions via randomized sampling on single core, multi core, and {GPU} architectures.
arXiv preprint, arXiv:1502.05366v3 (2016)

\bibitem {Intel}
\protect{Intel Parallel Studio XE Cluster Edition for Linux}.
\url{https://software.intel.com/en-us/intel-parallel-studio-xe} (2016)

\bibitem {Phillips2000feret}
Phillips, P. J., Moon, H., Rizvi, S. A., Rauss, P. J.:
The \protect{FERET} evaluation methodology for face-recognition algorithms.
IEEE Trans. Patt. Anal. Mach. Intell. 22, 1090--1104 (2000)



% additional reference suggested by IJCAI reviewer:

\end{thebibliography}
\end{document}